\documentclass[sigconf, nonacm, authorversion]{aamas}

\newif\iffullpaper\fullpapertrue 

\usepackage{balance} 

\setcopyright{ifaamas}
\acmConference[AAMAS '26]{Proc.\@ of the 25th International Conference
on Autonomous Agents and Multiagent Systems (AAMAS 2026)}{May 25 -- 29, 2026}
{Paphos, Cyprus}{C.~Amato, L.~Dennis, V.~Mascardi, J.~Thangarajah (eds.)}
\copyrightyear{2026}
\acmYear{2026}
\acmDOI{}
\acmPrice{}
\acmISBN{}

\usepackage[dvipsnames]{xcolor}
\usepackage{natbib}
\usepackage{amsmath,amsfonts} 
\usepackage{bm,xspace}
\usepackage{url} 
\usepackage{listings}
\usepackage{soul}  
\usepackage[utf8]{inputenc} 
\usepackage{booktabs}  
\usepackage{macros} 
\usepackage{multirow}
\usepackage[utf8]{inputenc}

\usepackage{pdfpages}

\usepackage{tikz}
\usetikzlibrary{calc,arrows,shapes,backgrounds}
\usetikzlibrary{arrows,shapes,decorations,automata,backgrounds,positioning}
\usetikzlibrary{shadows,shadings,shapes.symbols}
\usetikzlibrary{patterns.meta}

\tikzset{
    double color fill/.code 2 args={
        \pgfdeclareverticalshading[%
        tikz@axis@top,tikz@axis@middle,tikz@axis@bottom%
        ]{diagonalfill}{100bp}{%
            color(0bp)=(tikz@axis@bottom);
            color(50bp)=(tikz@axis@bottom);
            color(50bp)=(tikz@axis@middle);
            color(50bp)=(tikz@axis@top);
            color(100bp)=(tikz@axis@top)
        }
        \tikzset{shade, left color=#1, right color=#2, shading=diagonalfill}
    }
}

\tikzstyle{player}=[state,draw,rounded rectangle,align=center]
\tikzstyle{widget}=[draw=red,rectangle, rounded rectangle=10pt,dashed,minimum size=6mm,fill=yellow]
\tikzset{every loop/.style={looseness=7}}

\usepackage{nicefrac}

\usepackage{thm-restate}

\theoremstyle{remark}
\newtheorem{remark}{Remark}

\acmSubmissionID{21 (AAAI Track)}

\title{Verification of Robust Multi-Agent Systems}
\iffullpaper 
\thanks{This is an extended version of the paper with the same title that will appear in AAMAS 2026, which contains technical appendices with proof details.}
\else
\subtitle{AAAI Track}
\fi 

\author{Raphaël Berthon}
\affiliation{
  \institution{RWTH Aachen University}
  \city{Aachen}
  \country{Germany}}
\affiliation{
  \institution{Université Paris-Saclay, CNRS, ENS Paris-Saclay}
  \city{Gif-sur-Yvette}
  \country{France}}
\email{rberthon@lmf.cnrs.fr}

\author{Joost-Pieter Katoen}
\affiliation{
  \institution{RWTH Aachen University}
  \city{Aachen}
  \country{Germany}}
\email{katoen@cs.rwth-aachen.de}

\author{Munyque Mittelmann}
\affiliation{
  \institution{LIPN, CNRS, Université Sorbonne Paris Nord}
  \city{Villetaneuse}
  \country{France}}
\email{mittelmann@lipn.univ-paris13.fr}

\author{Aniello Murano}
\affiliation{
  \institution{University of Naples Federico II}
  \city{Naples}
  \country{Italy}}
\email{aniello.murano@unina.it}

\begin{abstract}

Stochastic multi-agent systems are a central modeling framework for autonomous controllers, communication protocols, and cyber-physical infrastructures. In many such systems, however, transition probabilities are only estimated from data and may therefore be partially unknown or subject to perturbations. 
In this paper, we study the verification of robust strategies in stochastic multi-agent systems with imperfect information, in which coalitions must satisfy a temporal specification while dealing with uncertain system transitions, partial observation, and adversarial agents. 
By focusing on bounded-memory strategies, 
we introduce a robust variant of the model-checking problem for a probabilistic, observation-based extension of Alternating-time Temporal Logic. We characterize the complexity of this problem under different notions of perturbation, thereby clarifying the computational cost of robustness in stochastic multi-agent verification and supporting the use of bounded-memory strategies in uncertain environments.

\end{abstract}


\newcommand{\BibTeX}{\rm B\kern-.05em{\sc i\kern-.025em b}\kern-.08em\TeX}

\begin{document}


\pagestyle{fancy}
\fancyhead{}


\maketitle 


\section{Introduction}
\label{sec:introduction}

AI-based technologies have been widely adopted in a variety of fields, including robotics, autonomous agents, and computer vision. 
The ubiquity 
of such technologies highlights the challenges of ensuring the trustworthiness, correctness, and robustness of computational systems. 
Formal verification
provides principled 
approaches 
to address these challenges.
In particular model checking \citep{clarke2018handbook} offers fully automated techniques 
in which systems are represented as labeled transition models, and their correct behaviors are specified by temporal logic formulas.  

Multi-Agent Systems (MAS) are systems composed of multiple autonomous components and interacting in a shared environment \citep{wool2009-MAS-book}. 
The evolution of a MAS is determined by the behavior of the agents.
Alternating-time Temporal Logic  (\ATL)  ~\citep{AlurHK02} is a prominent formal method to reason about strategies in MAS. Besides the traditional temporal modalities, \ATL contains strategic modalities to express 
how coalitions of agents can cooperate or compete to achieve their objectives.
%
Several aspects of MAS are inherently uncertain, due to both random events and the unpredictable behavior of agents. Such uncertainty can be quantified from experiments or past observations and captured by stochastic models, such as Markov decision processes (MDPs) and stochastic MAS. Probabilistic \ATL (\PATL) \citep{chen2007probabilistic} extends \ATL to the probabilistic setting, allowing reasoning about randomized strategic abilities of agents interacting
in a system with stochastic transitions. On the other hand, uncertainty in MAS may also originate from agents’ partial observability
of the system. In this imperfect-information setting, model checking strategic abilities with perfect recall is undecidable, even for deterministic MAS \citep{DBLP:journals/corr/abs-1102-4225}, which motivates a careful choice of restricted yet expressive classes of strategies \cite{Berthon21,journals-BLMR20}. 

\ignore{
Several aspects of MAS, such as the unpredictable behavior of agents and the occurrence of random phenomena, are uncertain.  These aspects can be measured based on experiments or past observations and represented with stochastic models, such as Markov decision processes (MDP) and stochastic MAS. Probabilistic \ATL (\PATL) \citep{chen2007probabilistic} extends \ATL to the probabilistic setting, allowing reasoning about randomized strategic abilities of agents interacting in a system with stochastic transitions.
On the other hand, uncertainty in MAS may also originate from agents' partial observability of the system.   
Model checking strategic abilities of agents with perfect recall under imperfect information entails undecidability, even when restricted to deterministic MAS \citep{DBLP:journals/corr/abs-1102-4225}. 
}

 Although the frequency of random events can be measured,  the precise probabilities are often unknown, and the system may face perturbations. For example, different weather centers often provide different precipitation probabilities for the same region and time. 
Another example is model-based reinforcement learning, where agents estimate the agent-environment interaction model (e.g., an MDP \citep{moerland2023model}) from data. 
Since the model is learned from their interaction with the environment, its transitions are susceptible to  
errors. 
Strategic behavior in such settings is difficult, as it requires dealing with uncertainty in the system transitions while interacting with other agents, who may be cooperative or adversarial. 

 

 In this paper, we investigate verification methods for \textit{robust strategies}, i.e., strategies ensuring a temporal specification despite perturbations in transition probabilities. We also consider whether \textit{coalitional} strategies exist that satisfy $\PATL$-like specifications. Instead of fixed transition probabilities after joint actions, we consider stochastic MAS with probability ranges within a given perturbation bound. The resulting setting captures both coalitional strategic uncertainty and uncertainty on exact transition probabilities.
 
To the best of our knowledge, this is the first work on stochastic MAS with ATL objectives where the transition probabilities are not fixed. This direction was not explored, including in the perfect information case. 
We consider three definitions of robustness. The first, using $\varepsilon$-perturbations, 
follows the classical setting in which every transition probability in the stochastic system may deviate by at most $\varepsilon$ from a nominal value ~\citep{DBLP:journals/ipl/ChenHK13}.
The other two notions are more involved and based on parametric systems, where transition probabilities are represented as rational functions over a finite set of \emph{parameters}.
We study both the case of a fixed number of parameters and the case of an arbitrary (unbounded) number of parameters. This separation is crucial because the arbitrary parameter case typically leads to significant jumps in complexity. 

In our setting, agents' strategies depend on their observations of the system, which captures imperfect information. Although recent work has provided decidable model-checking results for \PATL under imperfect information and memoryless strategies \citep{BelardinelliJMM23,BelardinelliJMM24}, we also allow for strategies with bounded memory.  
We consider strategies based on automata, which are widely used in verification~\citep{DBLP:books/daglib/0020348} and offer more expressivity in relation to standard bounded memory considered in ATL. We illustrate formally this improvement in Proposition~\ref{prop:strat_gen}, where we show that our automata-based strategies strictly generalize classical $k$-recall strategies.
At the same time, 
this additional expressiveness does not increase the worst-case complexity of model checking in our setting.

This setting is interesting from a practical point of view, as full memory is not always feasible. In fact, in real-world applications, such as economics, robotics, or AI, agents 
are often limited by their computational resources, processing power, and time constraints. Analogous limitations may occur with humans who interact with the system, who can only retain and process a finite amount of information \citep{natStrategy, jamroga2020natural}.  
Observation-based strategies with bounded memory allow handling uncertainty while being computationally feasible and avoiding the extremes of either oversimplification (memoryless strategies) or impracticality (unbounded memory). Having a richer representation for memory is particularly important under partial observation, where bounded memory is one of the few ways to obtain decidability \citep{CLMM18}. 
Therefore, bounded memory strikes a balance between flexibility and efficiency, 
which is important in practical scenarios. 

\begin{table}[h] 
\centering\noindent%
{\small
\begin{tabular}{@{}llll@{}}
\hline
  & Universal
  reachability  & Observation-based  \\ 
    &   (memoryful)&  \PATL  (bounded-memory) \\ \hline
 $\varepsilon$-perturb.& in $\PTIME$ (Thm.~\ref{prop:check-reach-eps}$^{*}$)~\citep{DBLP:journals/ipl/ChenHK13} & in $\Sigma_2^{\PTIME}$ (Thm.~\ref{prop:rob-patl}$^{**}$) \\
 Fix param. & in $\NPInter$ (Thm.~\ref{prop:check-reach-fixed}$^{*}$) & in $\Sigma_2^{\PTIME}$ (Thm.~\ref{prop:rob-patl-bound}$^{**}$) \\
 Unbounded  & in $\forall\Reals$ (Thm.~\ref{prop:check-reach-gen}$^{*}$) & in $\Sigma_3^{\mathbb{R}}$ (Thm.~\ref{prop:rob-patl-gen}$^{*}$) \\ \hline 
\end{tabular}
}  
\caption{Summary of the results. 
Results labeled with $^*$ allow randomized strategies. Results labeled with $^{**}$ consider deterministic strategies for the coalition but allow randomized strategies outside the coalition. 
}  
\label{tab:mcheck-complexity}
\vspace{-0.5cm}
\end{table}

 \paragraph{Contribution. } 
  This paper provides complexity results for obtaining robust strategies with bounded memory in uncertain stochastic MAS. We show our definition of bounded memory is strictly more general than recalling the $b$ last states visited.  
We consider the logic $\PATL$  for reasoning about observation-based strategies with bounded memory. 
We then define a variant of the model-checking problem that allows for possible perturbations and determine its complexity for specifications in $\PATL$. 
Our approach has two main advantages. 
 On one side, bounded memory is a trade-off between computational cost and retaining information. 
 On the other hand, we can guarantee that strategies are resilient to (bounded)  perturbations of the models. We provide novel results with reasonable complexity for model checking, going from ``simple'' reachability objectives to more complex ones describing coalitional behavior. 
 Whereas our reachability results (Theorems~\ref{prop:check-reach-eps},~\ref{prop:check-reach-fixed},~\ref{prop:check-reach-gen}) are close to the existing literature, they serve as essential building blocks for our PATL results (Theorems~\ref{prop:rob-patl},~\ref{prop:rob-patl-bound},~\ref{prop:rob-patl-gen}). These latter results and their proofs are entirely new and technically demanding because they require combining MAS and parametric model checking.  \iffullpaper \else Omitted proofs and technical details can be found in \textcolor{red}{CITE}. \fi
 
 Instead of representing the memory of strategies with the $n$ last states visited, we use automata with $n$ states,  increasing the strategic capabilities. While randomized strategies are even more expressive, we obtain better complexity results for deterministic strategies.  We consider different kinds of perturbations: perturbations represented by parameters that can be common to different actions and thus model dependencies, and two subcases: (i) a fixed number of parameters, and (ii)  possible perturbations that can take any value within an interval of $\varepsilon$. Table \ref{tab:mcheck-complexity} summarizes our main results. We mainly focus on membership results. Some hardness results may be difficult to obtain even for basic cases. In particular, universal reachability with unbounded parameters is the dual of existential reachability for which $\exists\Reals$-hardness is a known open problem~\citep{DBLP:journals/jcss/JungesK0W21}. Going further to $\Sigma_3^{\mathbb{R}}$, hardness results in the hierarchy of the reals is known to be particularly challenging~\cite{DBLP:journals/mst/SchaeferS24}. 
 \iffullpaper We also discuss how different approaches fail at giving a $\Sigma_2^{\PTIME}$-hardness result in the appendix. 
 \else 
 We also discuss how different approaches fail at giving a $\Sigma_2^{\PTIME}$-hardness result in \textcolor{red}{CITE}. 
 \fi 
 Concerning the case with a fixed number of parameters, no $\NPInter$-hard problem is known, and the existence of such a hard problem would imply $\PTIME\neq \NP$. 

\section{Related work} 
\label{sec:rw}
Our work is related to the research on probabilistic logics for MAS, interval and parametric Markov Decision Processes (MDPs), and synthesis of robust strategies. 

\paragraph{Verification of Stochastic MAS}
The verification of stochastic MAS against specifications given in probabilistic logics has been widely studied, including with specifications in probabilistic \ATL (\PATL) \citep{chen2007probabilistic},   Probabilistic Alternating-Time $\mu$-Calculus \citep{song2019probabilistic} and Probabilistic Strategy Logic (\PSL) \citep{aminof2019probabilistic}.
Verification of concurrent stochastic games has also been implemented using the PRISM model checker \citep{Kwiatkowska2022,kwiatkowska2022probabilistic}. 
Recent work studied the model checking problem for \PATL under imperfect information and memoryless strategies for the proponent coalition, for both deterministic  \citep{BelardinelliJMM23} and mixed strategies \citep{BelardinelliJMM24}. 
Finally,  \citep{BerthonKMM24}  considered a variant of \PATL with probabilistic natural strategies.
While natural strategies also capture bounded memory, directly applying the methods proposed in this paper would require representing them as automata, leading to an exponential blow-up in their representation. 

\paragraph{Interval MDPs.}
Interval Markov chains (IMCs) generalize Markov chains with interval-valued transition probabilities: strategies must hold for any system whose transition set is within the interval.  
Similar to our approach, IMCs are a   modeling tool for probabilistic systems with uncertainty  of the exact transition probabilities \citep{hahn2019interval}.
The model checking of formulas in Linear Temporal Logic (\LTL) over IMCs is in \EXPSPACE and \PSPACE-hard~\citep{DBLP:conf/tacas/BenediktLW13}.
In the generalization to Interval MDPs (IMDPs), model checking 
Probabilistic Computation Tree Logic (\PCTL) has been proved polynomial~\citep{DBLP:journals/ipl/ChenHK13} by using the ellipsoid method. However, this approach cannot be easily and directly generalized to observation-based strategies. 
The problem of finding robust randomized strategies for multiple objectives 
in IMDPs was shown to be \PSPACE-hard \citep{hahn2019interval}. 
Another approach for the verification of \PCTL properties in uncertain MDPs considers convex uncertainty sets as a generalization of intervals \citep{PuggelliLSS13}.
Similarly, ~\citep{ijcai2024p741} consider MDPs where each transition has an uncertainty set, with mean-payoff objectives, and gives an $\NPInter$ membership result.

\paragraph{Parametric MDPs.}
A way to generalize uncertainty and robustness is by considering parameters. In particular, transition probabilities can be represented as equations over a given set of parameters. While IMCs consider intervals of transition probabilities, 
the transition probabilities of Parametric MCs are given by polynomials with rational coefficients over a fixed set of real-valued parameters. This generalizes IMCs, as dependencies between transition probabilities can occur. 
The model checking problem for
parametric Markov chains with specifications in \PCTL was studied in \citep{baier2020parametric}. 
In that work, the authors have shown that, in the univariate case,  the existential \PCTL model checking problem is \NP-complete, and identified fragments of \PCTL where model checking is solvable in polynomial time. 
On parametric MDPs, instead, even the simplest cases become challenging~\citep{DBLP:journals/jcss/JungesK0W21}: with a fixed number of parameters, any reachability objective can be done in $\NP$, but with an arbitrary number of parameters, most problems become $\exists\Reals$-complete. 
Our separation between fixed and unbounded perturbations takes its root in these results. 
A related problem is to compute an instantiation of the unspecified parameters in a parametric MDP such that the resulting MDP satisfies a given temporal logic specification. This problem was studied in \citep{CubuktepeJJKT22} using convex optimization techniques. 

In general, the problem of separating satisfying from violating regions is hard on parametric MDPs, since these are bounded by non-linear functions~\cite{DBLP:journals/fmsd/JungesAHJKQV24}.
Parametric MDPs have been used to synthesize strategies in a two-player partially observable stochastic game framework \citep{abs-1810-00092}. Such a framework generalizes partially observable MDPs (POMDP) to two agents, where only one of them  has full observability 
of the system.
The synthesis approach involves a reduction to a series of synthesis problems for parametric MDPs combined with mixed-integer linear
programming.  

\paragraph{Robust Strategy Synthesis.} 
The synthesis of robust memoryless policies for uncertain POMDPs with reachability and expected cost objectives has been investigated in \citep{Suilen0CT20}. Similar to IMDPs, uncertain POMDPs consider a set of probability 
intervals. 
The proposed solution uses convex optimization and the authors have shown that the problem is \NP-hard for memoryless behavioral strategies. 
This work was extended with finite-state policies, which maintain the same lower bound complexity  \citep{cubuktepe2021robust}.
The problem of synthesizing robust strategies for   MDPs with ellipsoidal uncertainty
 has been studied alongside  \PCTL specifications and applied to compute risk-limiting strategies for energy pricing \citep{PuggelliSS14}. 
Experimental evaluations show promising results on s-rectangular robust MDPs~\cite{DBLP:journals/mor/WiesemannKR13,DBLP:conf/icml/KumarWLM24}, but it is still too early to reasonably implement these techniques for multi-agent systems.

The robustness of MAS has also been investigated from other perspectives.  
Robust \ATL \citep{MuranoNZ23} is an 
extension of \ATL to express the ability of a coalition of agents to tolerate the violation of an environment assumption in a purely deterministic setting. The environment assumption is described as a temporal logic property. A violation means the assumption's truth value may be ``false'' in some moments of the system execution. Recently, this notion was extended to the probabilistic setting \citep{zimmermann2025robust}. 
While we consider verification under model perturbations,  \citet{yu2024} has studied the case in which the agent's actions may fail with a certain probability. 




\section{Preliminaries}
\label{sec:preliminars}

In this paper, we fix finite non-empty sets of agents $\Ag$, actions $\Act$,
and atomic propositions $\APf$. 
We write $\profile{\act}$ for a tuple of actions $(\act_\ag)_{\ag\in\Ag}$, one for each agent, and such tuples are called \emph{action profiles}.
Given an action profile $\profile{\act}$ and $\coalition\subseteq\Ag$, we let $\act_\coalition$ be the actions of agents in  $\coalition$, and $\profile{\act}_{-\coalition}$ be $(\act_\agb)_{\agb\not \in \coalition}$. Similarly, we let $\Ag_{-\coalition}=\Ag\setminus\coalition$. We denote by $\mathbb{N}_+$ the set $\{1,2,\ldots\}$.




\paragraph{Distributions. } Let $X$ be a finite non-empty set. A \emph{(probability) distribution} over $X$ is a function $\distribution:X \to [0,1]$ such that $\sum_{x \in X} \distribution(x) = 1$. Let $\Dist(X)$ be the set of distributions over $X$. We write $x \in \distribution$ for $\distribution(x) > 0$. 
If $\distribution(x) = 1$ for some element $x \in X$, then $\distribution$ is a \emph{point (a.k.a. Dirac) distribution}. 
If, for $i\in I$, $\distribution_i$ is a distribution over $X_i$, then, writing $X = \prod_{i\in I} X_i$, the \emph{product distribution} of the $\distribution_i$ is the distribution $\distribution:X \to [0,1]$ defined by $\distribution(x) = \prod_{i\in I} \distribution_i(x_i)$.

\paragraph{Stochastic Systems. } 
  A \emph{stochastic multi-agent system} (or simply \emph{system})
  $\System$ is a tuple 
  $ ( \setpos, \legal, 
  \trans, \vval,(\funobs_{\ag})_{ \ag\in \Ag})$ where 
     (i) $\setpos$ is a finite non-empty set of \emph{states};
      (ii) $\legal: \setpos \times \Ag \to 2^\Act\setminus\{\emptyset\}$ is a \emph{legality function} defining the available actions for each agent in each state, we write $\profile{\legal(\pos)}$ for the tuple $(\legal(\pos, \ag))_{\ag\in\Ag}$; 
      (iii) for each state $\pos \in \setpos$ and each  action $\mov \in \profile{\legal(\pos)}$, the \emph{stochastic transition function} $\trans$ gives the 
      probability $\trans(\pos, \mov)(s')$ of a transition from state $\pos$ for all $\pos' \in \setpos$ if each player $\ag \in \Ag$ plays the action $\mova$, and remark that $\trans(\pos, \mov)\in\Dist(\setpos)$; and
      (iv) $\vval:\setpos \to 2^{\APf}$ is a \emph{labeling function}.
      (v) for each $\ag\in \Ag$, $\funobs_{\ag}$ is an observation function $\funobs:\setpos\to O_{\ag}$ for some observation set $O_{\ag}$. 

For each state $\pos \in \setpos$ and joint action $\mov \in \prod_{\ag \in \Ag} \legal(\pos,\ag)$, we assume 
  that there is a state $\pos'\in\setpos$ such that $\trans(\pos, \mov)(\pos')$ is non-zero, that is every state has a successor state from a legal action, hence $\mov\in \legal(\pos,\ag)$. 

\begin{example}[River] 
\label{ex:river}

Let us consider a system $\CGS_{river}$ with two companies 
sharing the usage of a river. 
At every step, each company has two available actions: discharge wastewater directly into the river (action $d$) or treat it before discharging it into the river (action $t$).
Atomic propositions state whether the river's water quality has reached low  (proposition $low$) or high levels (proposition $high$). The system is shown in Figure \ref{fig:river}.  The propositions $low$ and $high$ are true only in state $q_{low}$ and $q_{high}$, resp.

 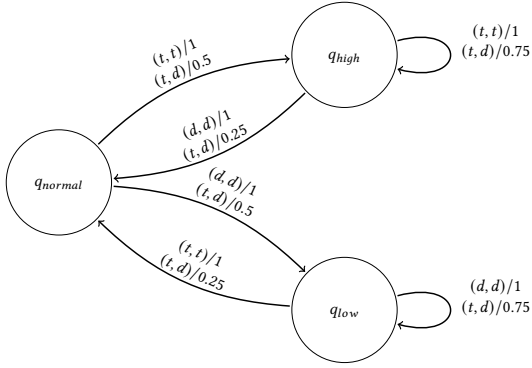
\begin{figure}[ht] 
     \centering
  
     \scalebox{.7}{
     \begin{tikzpicture}[shorten >= 1pt, shorten <= 1pt,  auto]
  \tikzstyle{rond}=[circle,draw=black,minimum size  = 2cm] 
 \tikzstyle{wina}=[fill=green!20]
  \tikzstyle{winb}=[fill=red!10] 
   \tikzstyle{label}=[sloped,shift={(0,-.1cm)}]
  
  \node[rond,fill=White] (q0) { \begin{tabular}{c} $q_{normal}$  \end{tabular}};
  \node[rond,fill=White] (q1) [ above right= 1cm and 4 cm  of q0] {\begin{tabular}{c} $q_{high}$   \end{tabular}}; 
  \node[rond,fill=White] (q2) [ below right= 1cm and 4 cm  of q0] {\begin{tabular}{c} $q_{low}$  \end{tabular}};

 \path[->,thick,bend left= 20] (q0) edge  node [label,pos=.5]    {{ 
 \begin{tabular}{c} $(t,t)/1$ \\  $(t,d)/0.5 $  \end{tabular}
}}    (q1);

 \path[->,thick,bend left= 20] (q0) edge  node [label,pos=.5]    {{ 
 \begin{tabular}{c} $(d,d)/1$ \\  $(t,d)/0.5 $  \end{tabular}
}}    (q2);

 \path[->,thick,bend left  = 20] (q1) edge  node [label,pos=.5]    {{ 
 \begin{tabular}{c} $(d,d)/1$ \\  $(t,d)/0.25 $  \end{tabular}
}}    (q0);

 \path[->,thick,bend left= 20] (q2) edge  node [label,pos=.5]    {{ 
 \begin{tabular}{c} $(t,t)/1$ \\  $(t,d)/0.25 $  \end{tabular}
}}    (q0);



\path [->,thick] (q1) edge[loop right]  node  [pos=.3]   {{ \begin{tabular}{c} $(t,t)/1$ \\  $(t,d)/0.75 $  \end{tabular}}}  (q1);


\path [->,thick] (q2) edge[loop right]  node  [pos=.3]   {{ \begin{tabular}{c} $(d,d)/1$ \\  $(t,d)/0.75 $  \end{tabular} }}  (q2);

\end{tikzpicture}
}
     \caption{System $\CGS_{river}$ representing the interaction between two companies. 
     The transition function is written in terms of labels on the arrows. Each transition from state $q$ to $q'$ is annotated with one or more labels of the form $(x, y)/z$ where $(x, y)$ denote the joint action of the companies in state $q$, and z denotes the probability of arriving at the state $q'$. The transition probabilities for the joint action $(d,t)$ are the same as for $(t,d)$, and are thus omitted. 
     } 
      \label{fig:river}
 \end{figure}

In state $q_{normal}$, no proposition is true, representing that the water quality is normal. If both companies discharge the wastewater, the water quality is guaranteed to decrease (from high to normal, and normal to low). 
Similarly, if both companies treat the water, the quality will increase (from low to normal, and normal to high). When only one company treats the water while the other discharges the wastewater, the effect is not deterministic. If the level was normal, the quality may increase or decrease with a probability of 50\%. If the quality is high (similarly, low), the quality may be maintained with a probability of 75\%  or decrease (resp. increase) with a probability of 25\%.

\end{example}

\paragraph{Markov Models. } A \emph{partially observable Markov decision process} (POMDP) is a one-player stochastic system, given with an observation function for every agent. When all observation functions are injective, the system is fully observable and we have a \emph{Markov decision process} (MDP) where we omit the observations. 

A \emph{Markov chain} is an MDP with $|\legal(\pos)| = 1$ for every $\pos\in\setpos$. We often represent a Markov chain $M$ as a tuple $(\setpos,p)$ where $\setpos$ is a countable non-empty set of states and $p \in \setpos\to\Dist(\setpos)$ the transition matrix. 

\paragraph{Plays. } 
A \emph{play} or path in a  system $\System$ is an infinite sequence $\iplay=\pos_0 \pos_1 \cdots\in\setpos^{\omega}$ of states
such that there exists a sequence $\mov_0 \mov_1 \cdots$ of joint-actions such that $\mov_i \in \legal(\pos_{i})$ and  $\pos_{i+1} \in \trans(\pos_i,\mov_i)$ (\ie, $\trans(\pos_i,\mov_i)(\pos_{i+1} )>0$) for every $i \geq 0$.
We write $\iplay_i$ for $\pos_i$, 
$\iplay_{\geq i}$ for the suffix of
$\iplay$ starting at position $i$. 
Finite paths are called \emph{histories}, and the set of all histories is denoted $\History$. We write $\last(\history)$ for the last state of a history $\history$ and $len(\history)$ for the size of $\history$. 
We extend observation functions $\funobs:\setpos\to O_{\ag}$ to paths with $\funobs:\setpos^{\omega}\to O_{\ag}^{\omega}$.



\section{General Strategies}

In what follows, we introduce a general class of observation-based strategies, both for infinite and bounded memory. 

\begin{definition}[General Bounded-Memory Strategies.]  
Let $\System= ( \setpos, \legal, \allowbreak \trans,  \vval, (\funobs_{\ag})_{ \ag\in \Ag})$ be a stochastic system with set of actions $\Act$ and set of observations $O$. In the most general case strategies for agent $\ag$ are unbounded functions $\sigma:\funobs_{\ag}(\History)\rightarrow\distribution(\Act)$. 
\end{definition}

We write $\setstrat_{\ag}$ for the set of unbounded strategies for agent $a$. 
We use Definition $10.97$ of~\citep{DBLP:books/daglib/0020348} and extend it as in~\citep{DBLP:journals/corr/abs-1006-1404} with randomization and partial observations. Strategy $\sigma$ is a \emph{general randomized strategy} with \emph{bounded recall} if it can be represented as a scheduler $(Q,act,\Delta,start)$ 
where $Q$ is a possibly infinite set of modes (or memory states),  
$\Delta: Q\times O \to \Dist(Q)$ is a randomized transition function, 
$act: Q\times O \to \Dist(\Act)$ randomly selects the next action,
and $start: O \to \Dist(Q)$ randomly selects a starting mode for the observation $o$. 
As in~\citep{DBLP:journals/corr/abs-1006-1404}, a strategy has infinite memory if $b = \infty$, finite memory $b\in \mathbb{N}_+$ if $|Q| =b$ and is \emph{memoryless} if $|Q| = 1$, it is \emph{behavioral} if $act$ is randomized but all other distributions are singletons, it is \emph{mixed} if $start$ is randomized but all other distributions are singletons, and it is \emph{pure} (or \emph{deterministic}) if all distributions are singletons. 
We call a strategy \emph{observation-based} if $\funobs$ isn't injective. We show in \iffullpaper the appendix \else  \textcolor{red}{CITE} \fi
that this captures other definitions, such as in~\cite{BelardinelliJMM23,BelardinelliJMM24,Berthon21}. 
Our definition of  $\compvar$-bounded recall strategies uses a finite-state scheduler. This is more general than recalling the last $b$ states, as we show next (see Prop. ~\ref{prop:strat_gen}). 


First, we point out differences in our representation of MAS with Imperfect Information in relation to the approach used in \citet{BelardinelliJMM23,BelardinelliJMM24}. 
While we define stochastic systems 
with imperfect information as  $\System= ( \setpos, \legal, \trans,  \allowbreak \vval,(\funobs_{\ag})_{ \ag\in \Ag})$,
they 
instead consider $\System= ( \setpos, \legal, \trans,  \allowbreak \vval,(\sim_{\ag})_{ \ag\in \Ag})$ where $\sim_{\ag}$ is an equivalence relation of $\setpos\times\setpos$. These two definitions of imperfect information are equivalent, since for $\pos,\pos'\in\setpos$, we can take $\funobs_{\ag}(\pos) = \funobs_{\ag}(\pos')$ iff $\pos\sim_{\ag}\pos'$. However, our definitions make it simpler to define our observation-based strategies using schedulers, in particular, functions 
$\Delta: Q\times O \to \Dist(Q)$, 
$act: Q\times O \to \Dist(\Act)$ 
and $start: O \to \Dist(Q)$. 

Our definition of finite-memory strategies differs from the usual one in logics for strategic reasoning, as in \cite{Berthon21}, which considers \emph{$k$-sequential strategies}, defined as $\sigma:\funobs_{\ag}(O^k)\rightarrow\distribution(\Act)$ for $k\in\mathbb{N}$. Our scheduler-based finite-memory strategies are strictly more general than sequential strategies, formally:
\begin{restatable}{proposition}{propstratgen}\label{prop:strat_gen}
The following holds: 
    \begin{itemize}
        \item For every $k$-sequential strategy, there exists a scheduler-based finite-memory strategy playing the same actions given the same history.
        \item There exists an MDP $\System_g$ and a \LTL formula $\varphi_g$ such that some scheduler-based finite-memory strategy satisfies $\varphi_g$ on $\System_g$ with probability $1$, but no $k$-sequential strategy satisfies it, regardless of $k$.  
        \end{itemize}
\end{restatable}
\begin{proof}[Proof Sketch]
For the second point, consider the system $\System_g$ in Figure~\ref{fig:mdp_gen} and the formula $\varphi_g = \G ((blue\Rightarrow \G \F blue) \wedge (red\Rightarrow \G \F red))$.
After any history on $\System_g$, the scheduler of Figure~\ref{fig:control_gen} is in $q_b$ (resp. $q_r$) iff $q_1$ (resp. $q_2$) was visited, and then plays the matching action, ensuring $\varphi_g$ almost surely.
In contrast, any $k$-sequential strategy must, after $k$ steps in $q_3$, play the same action in $q_4$ regardless of whether $q_1$ or $q_2$ occurred, so $\varphi_g$ cannot hold almost surely.
\iffullpaper
The full proof is in in Appendix~\ref{appendix:strat_gen}.
\else
The full proof is in in \textcolor{red}{CITE}.
\fi
\end{proof}


 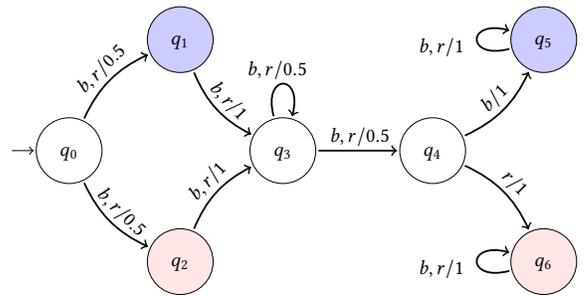
\begin{figure}[ht] 
     \centering
  
     \scalebox{.85}{
    \begin{tikzpicture}[shorten >= 1pt, shorten <= 1pt,  auto]
  \tikzstyle{rond}=[circle,draw=black,minimum size  = 1cm] 
 \tikzstyle{wina}=[fill=blue!20]
  \tikzstyle{winb}=[fill=red!10] 
   \tikzstyle{label}=[sloped,shift={(0,-.1cm)}]
  
  \node[rond,fill=White,initial,initial text=] (q0) { \begin{tabular}{c} $q_{0}$  \end{tabular}};
  \node[rond,wina] (q1) [ above right= 1cm and 1 cm  of q0] {\begin{tabular}{c} $q_{1}$   \end{tabular}}; 
  \node[rond,winb] (q2) [ below right= 1cm and 1 cm  of q0] {\begin{tabular}{c} $q_{2}$  \end{tabular}};
  \node[rond,fill=White] (q3) [ right= 2.3 cm  of q0]  { \begin{tabular}{c} $q_{3}$  \end{tabular}};
  \node[rond,fill=White] (q4) [ right= 1.3 cm  of q3]  { \begin{tabular}{c} $q_{4}$  \end{tabular}};
  \node[rond,wina] (q5) [ above right= 1cm and 1 cm  of q4] {\begin{tabular}{c} $q_{5}$   \end{tabular}}; 
  \node[rond,winb] (q6) [ below right= 1cm and 1 cm  of q4] {\begin{tabular}{c} $q_{6}$  \end{tabular}};

 \path[->,thick,bend left= 20] (q0) edge  node [label,pos=.5]    {{ 
 \begin{tabular}{c} $b,r/0.5$  \end{tabular}
}}    (q1);

 \path[->,thick,bend right= 20] (q0) edge  node [label,pos=.5]    {{ 
 \begin{tabular}{c} $b,r/0.5 $  \end{tabular}
}}    (q2);

 \path[->,thick,bend right= 20] (q1) edge  node [label,pos=.5]    {{ 
 \begin{tabular}{c} $b,r/1 $  \end{tabular}
}}    (q3);

 \path[->,thick,bend left= 20] (q2) edge  node [label,pos=.5]    {{ 
 \begin{tabular}{c} $b,r/1 $  \end{tabular}
}}    (q3);

\path [->,thick] (q3) edge[loop above]  node  [pos=.3]   {{ \begin{tabular}{c} $b,r/0.5$  \end{tabular}}}  (q3);

 \path[->,thick] (q3) edge  node [label,pos=.5]    {{ 
 \begin{tabular}{c} $b,r/0.5 $  \end{tabular}
}}    (q4);

 \path[->,thick,bend right= 20] (q4) edge  node [label,pos=.5]    {{ 
 \begin{tabular}{c} $b/1 $  \end{tabular}
}}    (q5);

 \path[->,thick,bend left= 20] (q4) edge  node [label,pos=.5]    {{ 
 \begin{tabular}{c} $r/1 $  \end{tabular}
}}    (q6);

\path [->,thick] (q5) edge[loop left]  node  [pos=.3]   {{ \begin{tabular}{c} $b,r/1$  \end{tabular}}}  (q5);

\path [->,thick] (q6) edge[loop left]  node  [pos=.3]   {{ \begin{tabular}{c} $b,r/1$  \end{tabular}}}  (q6);

\end{tikzpicture}
}
     \caption{An MDP $\System_g$ where scheduler-based strategies are more general than sequential strategies. 
    } 
      \label{fig:mdp_gen}
 \end{figure}

 \begin{figure}[ht] 
     \centering
  
     \scalebox{1.1}{
    \begin{tikzpicture}[shorten >= 1pt, shorten <= 1pt,  auto]
  \tikzstyle{rond}=[circle,draw=black,minimum size  = 1cm] 
 \tikzstyle{wina}=[fill=blue!20]
  \tikzstyle{winb}=[fill=red!10] 
   \tikzstyle{label}=[sloped,shift={(0,-.1cm)}]
  
  \node[rond,fill=White,initial,initial text=] (q0) { \begin{tabular}{c} $q_{b}$  \end{tabular}};
  \node[rond,fill=White] (q1) [right= 2cm of q0] {\begin{tabular}{c} $q_{r}$  \end{tabular}};

 \path[->,thick,bend left= 20] (q0) edge  node [label,pos=.5]    {{ 
 \begin{tabular}{c} $q_2/1$  \end{tabular}
}}    (q1);

 \path[->,thick,bend left= 20] (q1) edge  node [label,pos=.5]    {{ 
 \begin{tabular}{c} $q_1/1$  \end{tabular}
}}    (q0);

\path [->,thick] (q0) edge[loop above]  node  [pos=.3]   {{ \begin{tabular}{c} $q_0,q_1,q_3,q_4,q_5/1$  \end{tabular}}}  (q5);

\path [->,thick] (q1) edge[loop above]  node  [pos=.3]   {{ \begin{tabular}{c} $q_0,q_2,q_3,q_4,q_6/1$  \end{tabular}}}  (q6);

\end{tikzpicture}
}
     \caption{A two-state scheduler enforcing $\varphi_g$ on $\System_g$. 
    } 
      \label{fig:control_gen}
 \end{figure}
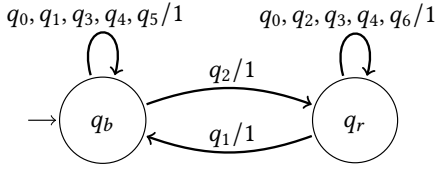

We highlight that our example uses a deterministic and perfect-information scheduler-based strategy. ~\citet{DBLP:journals/corr/abs-1006-1404} details how randomization leads to even more general strategies, hence our choice to call our strategies ``general'', or to simply refer to them as ``strategies''. We note that the two definitions coincide on infinite-memory strategies.

Observation-based and bounded recall strategies are not mutually exclusive.  
We take $\setstrat_{\ag,n}$ as the set of observation-based strategies with bounded recall $n$ for agent $\ag$ and  $\setstrat_n = \cup_{\ag \in \Ag} \setstrat_{\ag,n}$. 


\begin{remark}
     The techniques used later in this paper to establish the complexity results can be applied to both bounded recall and memoryless strategies, yielding identical complexity bounds. Since bounded recall captures the case of memoryless strategies and is more general, we will focus on \emph{bounded recall strategies}. 
 \end{remark}

\section{Probabilistic {\ATL} 
}
\label{sec:patl}


Next, we recall the syntax of Probabilistic alternating-time temporal logic (\PATL) \cite{BelardinelliJMM23}, 
defined as follows:


\begin{definition} \label{def:ATLF-syntax}
	The syntax of \PATL  is given by the grammar
	\begin{align*}   	
		\varphi ::= p \mid  \varphi \lor \varphi \mid \neg \varphi \mid \coop{\coalition}^{\bowtie d} (\X \varphi) & \mid \coop{\coalition}^{\bowtie d}(\varphi \until\varphi)	
  \\ &\mid  \coop{\coalition}^{\bowtie d} (\varphi \release \varphi)	
  \end{align*}
  where $p \in \APf$, 
  $\coalition \subseteq \Ag$, $d$ is a rational constant in $[0, 1]$, and $\bowtie \in  \{\leq, <, >, \geq\}$.  
\end{definition}

The intuitive reading of the operators is as follows: ``next'' $\X$, ``release'' $\release$ and ``until'' $\U$ are the standard temporal operators.  We make use of the usual syntactic sugar ${\F \varphi \colonequals \top \U \varphi}$ and ${\G \varphi \colonequals \bot \release  \varphi}$ for temporal operators.  $\coop{\coalition}^{\bowtie d}\varphi$ asserts that there exists an observation-based  strategy 
for the coalition $\coalition$ to collaboratively enforce $\varphi$  with a probability in relation $\bowtie$ with constant $d$.


Before presenting the semantics, we show how to define the probability space on outcomes. 

 \paragraph{Probability Space on Outcomes. } We define a
 strategy profile as $\profile\strat = (\strat_\ag)_{\ag \in \Ag}$, and its \emph{output} from state $\pos$ 
 is a play $\iplay$ that starts in state
 $\pos$ and is extended by letting each agent follow their strategies in $\profile\strat$, formally $\iplay_{0} =
 \pos$, and for every $\compvark \geq 0$ 
 there exists $\mov_{\compvark} \in
 (\strat_\ag(\iplay_{\leq \compvark}))_{\ag \in \Ag}$ such that $\iplay_{\compvark+1} \in
 \trans(\iplay_{\compvark},\mov_{\compvark})$. 

 The (potentially infinite)  set of outputs of a profile of  strategies   $\profile\strat$ and state $\pos$  is denoted $Output(\profile\strat,\pos)$.  
 
 A given 
 system 
 $\System$, strategy profile $\profile{\strat}$, and state 
 $\pos$ induce an infinite-state Markov chain
 $M_{\profile{\strat},\pos}$ whose states are finite prefixes of plays (i.e., histories)
 in 
 $Output(\profile{\strat},\pos)$. Given two histories $h$ and $hs'$ in $Output(\profile{\strat},\pos)$,  the  transition probabilities in  $M_{\profile{\strat},\pos}$ 
are defined as  $p(\history,\history\pos')=\sum_{\Pi_{a\in \Ag} \Act}
 \profile{\strat}(\mov)(\history) \times
 \trans(\last(\history),\mov)(\pos')$.  
 The Markov chain
 $M_{\profile{\strat},\pos}$ induces a canonical probability space on
 its set of infinite paths~\citep{kemeny1976stochastic}, 
 and thus also in $Output(\profile{\strat},\pos)$.~\footnote{This is a classic construction, see for instance
 ~\citep{clarke2018model,berthon2020alternating}.
 }

  Given a coalition strategy $\profile{\strat_\coalition} \in \prod_{\ag \in \coalition} \setstrat_{\ag,n}$  with memory bound $n$, 
the set of possible outcomes of $\profile{\strat_\coalition}$ from a state $\pos \in \setpos$ is the set $out_\coalition(\profile{\strat_\coalition},\pos) = \{Output((\profile{\strat_\coalition},\profile{\strat_{-\coalition}}),\pos) : \profile{\strat_{-\coalition}} \in \prod_{\ag \in \Ag_{-\coalition}} \setstrata \}$ of paths that the players in $\coalition$ enforce when they
follow the strategy $\profile{\strat_\coalition}$, namely, for each $\ag \in \Ag$,
player $\ag$ follows strategy $\strat_\ag$ in $\profile{\strat_\coalition}$. We use $\mu^{\profile{\strat_\coalition}}_\pos$ to range over the measures induced by  $out_\coalition(\profile{\strat_\coalition},\pos)$.

\begin{definition}
\PATL  formulas are interpreted in a stochastic system $\System$  
  and a path  $\iplay$. Let $b \in \mathbb{N}_+\cup\{\infty\}$. The semantics of \PATL with $b$-bounded memory, denoted by the satisfaction relation $\models_{\compvar}$, is defined  as follows:  
\begingroup
\allowdisplaybreaks
\begin{align*}
 \System,\iplay &\models_{\compvar} p & \text{ iff } & p \in v(\iplay_0)\\
 \System,\iplay &\models_{\compvar} \neg \varphi & \text{ iff } & \System,\iplay \not \models_{\compvar} \varphi \\
 \System,\iplay &\models_{\compvar} \varphi_1 \lor \varphi_2 & \text{ iff }&  
 \System,\iplay \models_{\compvar} \varphi_1  \text{ or } \System,\iplay \models_{\compvar} \varphi_2
 \\
\System, \iplay&\models_{\compvar} \coop{\coalition}^{\bowtie d} \varphi & \text{ iff }  & 
\exists \profile{\strat_{\coalition}} \in \prod_{\ag \in\coalition} \setstrat_{\ag,b}  
\\ &&&  \text{ s.t. }  \forall \mu^{\profile{\strat_\coalition}}_{\iplay_0} \in out_\coalition(\profile{\strat_{\coalition}},\iplay_0)  
\text{, }
\\ &&& \mu^{\profile{\strat_\coalition}}_{\iplay_0}(\{\iplay' : \System,\iplay' \models_{\compvar} \varphi\}) \bowtie d\span\span\span
\\ 
 \System,\iplay &\models_{\compvar} \X \varphi & \text{ iff } & \System,\iplay_{\geq 1} \models_{\compvar} \varphi \\
\System, \iplay  & \models_{\compvar} \phi_1 \until \phi_2 & \text{ iff } &  \exists \compvark \geq 0 \text{ s.t. } \System,\iplay_{\geq \compvark} \models_{\compvar} \phi_2 \text{ and } 
\\ & & &
\forall j \in [0,\compvark).\, \,  \System,\iplay_{\geq j}\models_{\compvar} \phi_1 \\
\System, \iplay  & \models_{\compvar} \phi_1 \release \phi_2 & \text{ iff } &  \forall \compvark \geq 0 \text{ s.t. } \System,\iplay_{\geq \compvark} \models_{\compvar} \phi_2 \text{ or } 
\\ & & &
\exists j \in [0,\compvark).\, \,  \System,\iplay_{\geq j}\models_{\compvar} \phi_1 
\\
\end{align*} 
\endgroup
\end{definition}
\begin{remark}
    \PATL defined over systems with exactly one agent   corresponds to \PCTL over POMDPs with bounded-memory policies, and the literature we build on focuses almost exclusively on the fragment of \PCTL restricted to reachability objectives. Our objective is to handle all \PATL with bounded-memory in the MAS setting.  
\end{remark}

\begin{definition}
Given 
a system $\System$, state $s \in \setpos$, a memory bound $\compvar\in\mathbb{N}_+$  and   a formula $\varphi$ in 
\PATL, 
the \emph{model checking problem} for
\PATL consists of deciding
whether $\System, s \models_{\compvar} \varphi$. 
\end{definition}

\begin{proposition}[From~\cite{DBLP:conf/prima/HaoSLSGDL12}]
    The model checking problem for \PATL using infinite-memory strategies is undecidable.
\end{proposition}\begin{proof}
With infinite memory, the definition of scheduler-based strategies and classical sequential strategies coincide as unbounded functions $\sigma:\funobs_{\ag}(\History)\rightarrow\distribution(\Act)$. Hence we can directly apply the known result that \PATL with infinite memory and under partial observation is undecidable \cite{DBLP:conf/prima/HaoSLSGDL12}.     
\end{proof}



For the rest of the paper, we will assume $b \in \mathbb{N}_+$.




\section{Robustness}
\label{sec:robust}

\paragraph{Parametric Systems. }
In a parametric system, transition probabilities are replaced with equations over a finite set of variables $X$. We can thus consider whether a property holds for all valuations over $X$ that yield a well-defined set of probabilities. We use a definition similar to~\citep{DBLP:journals/jcss/JungesK0W21}.

  A \emph{parametric stochastic system} (or \emph{\param-System})
  $\System$ is a tuple 
  $ ( \setpos, \legal, X, 
  \trans, \vval,(\funobs_{\ag})_{ \ag\in \Ag}, \constraints)$ where 
  \begin{itemize}
      \item $\setpos$, $\legal$, $\vval$, and $(\funobs_{\ag})_{ \ag\in \Ag}$ are defined as for stochastic systems;
      \item $X$ is a finite set of real-valued \emph{parameters};
      \item   for each state $\pos \in \setpos$ and each  action $\mov \in \profile{\legal(\pos)}$, the \emph{parametric stochastic transition function} $\trans$ gives the (conditional) probability of a transition from state $\pos$ for all $\pos' \in \setpos$ if each player $\ag \in \Ag$ plays the action $\mova$. This probability is a rational function over $X$, formally $\trans(\pos, \mov)(s') \in \mathbb{Q}[X]$; and 
  \item $\constraints$ is a set of polynomial constraints of the form $A \bowtie B$ where $A,B\in \mathbb{Q}[X]$ and $\bowtie\in\{<,>,\leq,\geq, =\}$. Later on, we use these constraints to capture different kinds of perturbation. 
  \end{itemize} 
  




On a \param-system $\System$, a valuation $\Vval:X\to \mathbb{R}$ is \emph{well-defined} if replacing every $x\in X$ by $\Vval(x)$ yields a stochastic system $\System[\Vval]$. By extending the co-domain of $\Vval$ to rational functions over $X$, we can formally define well-defined valuations as requiring (i) probabilities to be non-negative: for each $\pos,\pos' \in \setpos$ and $\mov \in \profile{\legal(\pos)}$, we have $\Vval(\trans(\pos,\mov,\pos'))\geq 0$, and (ii) transitions to induce distributions: for each $\pos, \in \setpos$ and $\mov \in \profile{\legal(\pos)}$ we have $\Sigma_{\pos' \in\setpos}\Vval(\trans(\pos,\mov,\pos')) = 1$.

As above, \emph{parametric POMDPs} or \param-POMDPs are defined as \param-Systems with only one agent, \emph{parametric MDPs} or \param-MDPs as \param-POMDPs with perfect information, and \emph{parametric MCs} or \param-MCs as \param-MDPs with $|\legal(\pos)| = 1$ for every $\pos\in\setpos$.


\begin{example}[River (cont.)]
\label{ex:river2}

Suppose now that the river’s water quality is less stable than previously assumed. We remark on some dependency in this instability and represent it with parameters. Some transition probabilities now involve a perturbation $x$, and all these probabilities change together. 
We take the following family of perturbations parameterized by $x$ in the system $\CGS_{river}$ from Example \ref{ex:river}, that we call $\CGS^p_{river}$ and illustrated in Figure~\ref{fig:river3}. 

 \begin{figure}[hb] 
     \centering
  
     \scalebox{.7}{
     \begin{tikzpicture}[shorten >= 1pt, shorten <= 1pt,  auto]
  \tikzstyle{rond}=[circle,draw=black,minimum size  = 2cm] 
 \tikzstyle{wina}=[fill=green!20]
  \tikzstyle{winb}=[fill=red!10] 
   \tikzstyle{label}=[sloped,shift={(0,-.1cm)}]
  
  \node[rond,fill=White] (q0) { \begin{tabular}{c} $q_{normal}$  \end{tabular}};
  \node[rond,fill=White] (q1) [ above right= 1cm and 4 cm  of q0] {\begin{tabular}{c} $q_{high}$   \end{tabular}}; 
  \node[rond,fill=White] (q2) [ below right= 1cm and 4 cm  of q0] {\begin{tabular}{c} $q_{low}$  \end{tabular}};

 \path[->,thick,bend left= 20] (q0) edge  node [label,pos=.5]    {{ 
 \begin{tabular}{c} $(t,t)/1$ \\  $(t,d)/0.5+x $  \end{tabular}
}}    (q1);

 \path[->,thick,bend left= 20] (q0) edge  node [label,pos=.5]    {{ 
 \begin{tabular}{c} $(d,d)/1$ \\  $(t,d)/0.5-x $  \end{tabular}
}}    (q2);

 \path[->,thick,bend left  = 20] (q1) edge  node [label,pos=.5]    {{ 
 \begin{tabular}{c} $(d,d)/1$ \\  $(t,d)/0.25+x $  \end{tabular}
}}    (q0);

 \path[->,thick,bend left= 20] (q2) edge  node [label,pos=.5]    {{ 
 \begin{tabular}{c} $(t,t)/1$ \\  $(t,d)/0.25 $  \end{tabular}
}}    (q0);



\path [->,thick] (q1) edge[loop right]  node  [pos=.3]   {{ \begin{tabular}{c} $(t,t)/1$ \\  $(t,d)/0.75-x $  \end{tabular}}}  (q1);


\path [->,thick] (q2) edge[loop right]  node  [pos=.3]   {{ \begin{tabular}{c} $(d,d)/1$ \\  $(t,d)/0.75 $  \end{tabular} }}  (q2);

\end{tikzpicture}
}
     \caption{The parametric system $\CGS^p_{river}$ shows the case where we are unsure of some transitions, that may change together depending on $x\in[0,0.25]$.
     } 
      \label{fig:river3}
 \end{figure}
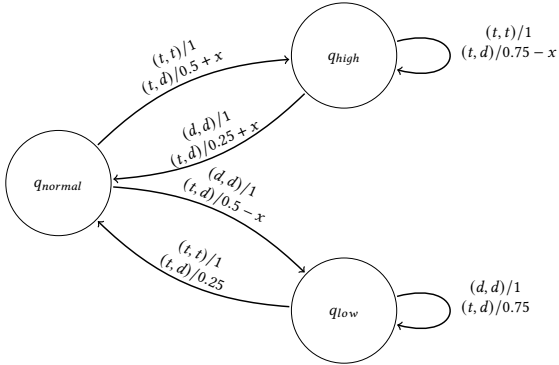

When in state $q_{normal}$, the probability that the water quality increases when only one company treats the water may be slightly higher than planned, and so the probability of going to  $q_{high}$ is $0.5+x$. As a consequence, the probability to go to $q_{low}$ under this situation is $0.5-x$, to compensate. 
At the same time, the probability of staying in $q_{high}$ with only one company treating water may be slightly lower than expected. but follows the same trend, hence it is $0.75-x$. Still, we can check that, as long as $x$ stays within $0$ and $0.25$, a company alone can always make sure to have at least probability $0.375$ of having a good water quality within two time steps by treating the water.

\end{example}

Following prior work ~\citep{hahn2019interval,Suilen0CT20,cubuktepe2021robust,ijcai2024p741}, the notion of robust strategies we consider requires dealing with uncertainty about the precise probabilities. Such uncertainty is modeled in terms of \textit{perturbations} that the system may face. Our first goal is to verify whether a strategy achieves a \textit{goal}, even though the system may face perturbations of  $\varepsilon\in[0,1]$. We start by introducing $\varepsilon$-approximated systems. 



Given a stochastic system $\System = ( \setpos, \legal, 
  \trans, \val,(\funobs_{\ag})_{ \ag\in \Ag})$, and $\varepsilon\in[0,1]$, we define the \emph{$\varepsilon$-approximated} parametric system $\System_{\varepsilon} =  ( \setpos, \legal, \allowbreak  X,  \trans', \val,(\funobs_{\ag})_{ \ag\in \Ag},\constraints)$ where for every $\pos,\pos'\in\setpos$ and $\mov \in \profile{\legal(\pos)}$, every transition probability $\trans'(\pos,\mov,\pos')\in X$ is a new variable. We add to $\constraints$ the constraints that $\trans'(\pos,\mov,\pos')-\trans(\pos,\mov,\pos') \leq \varepsilon$  and $\trans(\pos,\mov,\pos')-\trans'(\pos,\mov,\pos') \leq \varepsilon$, or equivalently  $|\trans(\pos,\mov,\pos')-\trans'(\pos,\mov,\pos')| \leq \varepsilon$. Hence, compared with $\System$, every probability in $\System_{\varepsilon}$ may have an error of up to $\varepsilon$.
  IMDPs are a subcase of $\varepsilon$-approximated parametric systems, more precisely IMDPs correspond to $\varepsilon$-approximated MDPs~\citep{DBLP:journals/ipl/ChenHK13}. 

  


\begin{remark}
    In general, parameters can introduce discontinuities in the structure of the graph, since edges can disappear when they probability become $0$. While it could change the complexity of the problem, this is not the case for the two results we use as building blocks in~\cite{DBLP:journals/ipl/ChenHK13,DBLP:journals/jcss/JungesK0W21}. Hence we have no need to check whether the graph structure is modified by the parameter valuations.
\end{remark}

\begin{example}[River (cont.)]
We now consider an interval perturbation $\epsilon=0.05$ in the system $\CGS_{river}$ from Example \ref{ex:river}. 
For each possible states $q$, $q'$ and joint action $(x,y)$, the $0.05$-approximated parametric system $\CGS_{0.05,river}$ allows for variations in the transitions probability that differ at most $0.05$ from the transitions in $\CGS_{river}$. Hence probability $0.5$ may be replaced by any value in $[0.45,0.55]$. One possible realisability of  $\CGS_{0.05,river}$ is illustrated in Figure \ref{fig:riverpert}. The obtained system, which we shall denote $\CGS_{river}'$, has some notable differences in relation to $\CGS_{river}$. First,  the transition from $q_{high}$ when both companies discharge wastewater 
is no longer deterministic. With high probability, the transition still leads to state $q_{normal}$ (that is, the water quality is decreased to normal levels), but there is a small chance that the water quality decreases to low levels, moving to state $q_{low}$ instead. 
Second, when only one company treats the water at state $q_{normal}$, it is more likely that the water quality will decrease (moving to state $q_{low}$) than to increase  (moving to state $q_{high}$). 
Clearly, 
it is more probable to reach state $q_{low}$ from the other two states in  $\CGS_{river}'$ than in the original system $\CGS_{river}$. Notice this is not the case for every realisability of $\CGS_{0.05,river}$, as it admitts other modifications in the  transition probabilities.

 \begin{figure}[ht] 
     \centering
  
     \scalebox{.7}{
    \begin{tikzpicture}[shorten >= 1pt, shorten <= 1pt,  auto]
  \tikzstyle{rond}=[circle,draw=black,minimum size  = 2cm] 
 \tikzstyle{wina}=[fill=green!20]
  \tikzstyle{winb}=[fill=red!10] 
   \tikzstyle{label}=[sloped,shift={(0,-.1cm)}]
  
  \node[rond,fill=White] (q0) { \begin{tabular}{c} $q_{normal}$  \end{tabular}};
  \node[rond,fill=White] (q1) [ above right= 1cm and 4 cm  of q0] {\begin{tabular}{c} $q_{high}$   \end{tabular}}; 
  \node[rond,fill=White] (q2) [ below right= 1cm and 4 cm  of q0] {\begin{tabular}{c} $q_{low}$  \end{tabular}};

 \path[->,thick,bend left= 20] (q0) edge  node [label,pos=.5]    {{ 
 \begin{tabular}{c} $(t,t)/1$ \\  $(t,d)/0.45 $  \end{tabular}
}}    (q1);

 \path[->,thick,bend left= 20] (q0) edge  node [label,pos=.5]    {{ 
 \begin{tabular}{c} $(d,d)/1$ \\  $(t,d)/0.55 $  \end{tabular}
}}    (q2);

 \path[->,thick,bend left  = 20] (q1) edge  node [label,pos=.5]    {{ 
 \begin{tabular}{c} $(d,d)/0.95$ \\  $(t,d)/0.25 $  \end{tabular}
}}    (q0);

 \path[->,thick,bend left= 20] (q2) edge  node [label,pos=.5]    {{ 
 \begin{tabular}{c} $(t,t)/1$ \\  $(t,d)/0.25 $  \end{tabular}
}}    (q0);

\path[->,thick, bend left= 20] (q1) edge  node [label,pos=.5]    {{  \begin{tabular}{c} $(d,d)/0.05$   \end{tabular}}}    (q2);


\path [->,thick] (q1) edge[loop right]  node  [pos=.3]   {{ \begin{tabular}{c} $(t,t)/1$ \\  $(t,d)/0.75 $  \end{tabular}}}  (q1);


\path [->,thick] (q2) edge[loop right]  node  [pos=.3]   {{ \begin{tabular}{c} $(d,d)/1$ \\  $(t,d)/0.75 $  \end{tabular} }}  (q2);

\end{tikzpicture}
}
     \caption{The system $\CGS_{river}'$, obtained from a realisability of the $0.05$-approximated system $\CGS_{0.05, river}$.
     } 
      \label{fig:riverpert}
 \end{figure}
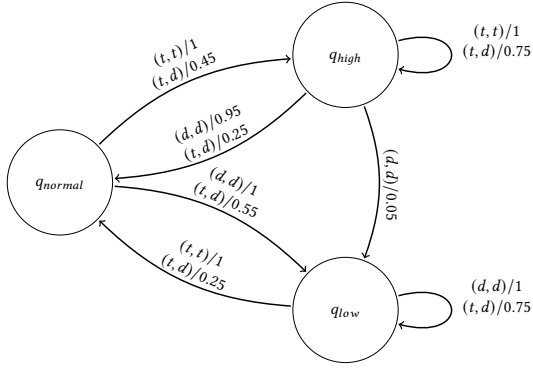

\end{example}

We now introduce the robust model checking problem for \PATL. 
\begin{definition}
Given 
a parametric system $\System$, state $s \in \setpos$, 
and  
  a formula $\varphi$ in 
\PATL, a memory-bound $\compvar$, 
\emph{the parametric model checking problem} for
\PATL consists of deciding
whether $\System[Val], s \models_{\compvar} \varphi$ for each well-defined valuation $Val$.

In the case where $\System$ is an $\varepsilon$-approximated system, we call it the  $\varepsilon$-robust model checking problem for
\PATL.
\end{definition}

We also consider the case where we assume the number of perturbed transitions is fixed. Indeed, in many cases, we can consider that only a few critical components have an uncertain behavior. As we will see, this assumption greatly helps to reduce the model checking complexity.

\section{Robust Model Checking of {\PATL}
}
\label{sec:modelchecking}



We start by the case in which the specification $\varphi$ defines a reachability objective. More precisely, we are interested in the case where a reachability specification $\varphi$ holds for all (arbitrary) strategies of the agents, formally:  
$\varphi = \coop{\emptyset}^{\bowtie d} \F R$, 
 where  
 $d$ is a rational constant in $[0, 1]$, $\bowtie \in  \{\leq, <, >, \geq\}$, and $R$ is a Boolean combination of atomic propositions. This is the same as considering whether a formula holds for all possible strategies in a parametric MDP. To obtain results on existential reachability, duality can be used, since there exists a strategy reaching a target $R$ with probability $>d$ iff we can prove that it is not the case that all strategies reach $R$ with probability $\leq d$. For more results on parametric reachability, see~\citep{DBLP:journals/jcss/JungesK0W21}.

 We start with $\varepsilon$-approximated MDPs. Since they are a specific case of interval MDPs, we can use the following result:

 \begin{restatable}{theorem}{checkreacheps}[From~\citep{DBLP:journals/ipl/ChenHK13}]
    \label{prop:check-reach-eps}
    The following problem is polynomial: Given an $\varepsilon$-approximated MDP $\System = ( \setpos, \legal, X,  \trans, \val,(\funobs_{\ag})_{ \ag\in \Ag}, \constraints)$, and a reachability objective $R$, determine whether for all $\compvar$-bounded strategies  and all well-defined parameters $X$, $R$ can be reached with probability $\bowtie d$ with $\bowtie\in\{<,>,\geq,\leq\}$.
\end{restatable}




For the more general case of parametric systems, we first consider the case in which we have a bound on the number of parameters, extending results on parametric MDPs from~\cite{DBLP:journals/jcss/JungesK0W21}. 

 \begin{restatable}{theorem}{checkreachfixed}
    \label{prop:check-reach-fixed} 
    The following problem is in $\NPInter$\!\!: 
    Given a \param-MDP $\System = ( \setpos, \legal, X,  \trans, \val,(\funobs_{\ag})_{ \ag\in \Ag}, C)$ where the number of parameters is less than a fixed $m$, and a reachability objective $R$, 
    determine whether for all  $\compvar$-bounded strategies and all well-defined parameters $X$, $R$ can be reached with probability $\bowtie d$. 
\end{restatable}
\begin{proof}
We first show that this problem is in $\coNP$. The dual problem, that is deciding if there exists a parameter and a strategy that does satisfy a reachability probability with probability $\bowtie d$, is in $\NP$ by Proposition 15 of~\citep{DBLP:journals/jcss/JungesK0W21}, hence our problem is in $\coNP$.

To show that our problem is in $\NP$, we first remark that we are trying to find a strategy satisfying a reachability property in an MDP: we restrict ourselves to memoryless strategies (by Lemma 10.102 of~\citep{DBLP:books/daglib/0020348}). We use the same technique as Theorem 12 of~\citep{DBLP:journals/jcss/JungesK0W21}: We guess a strategy, representing the worst possible case, either minimizing (for $\geq d$ and $> d$) or maximizing (for $\leq d$ and $< d$) the probability to reach $R$, and check that this worst-case strategy is still optimal. Using the Bellman equations, we translate this optimality constraint into an $\exists\Reals$ formula $\varphi$ with a fixed number of parameters. Since we are trying to decide whether for all well-defined valuations on $X$, $\varphi$ holds, we have a $\forall\exists\Reals$ formula with a fixed amount of parameters, which can be checked in $\PTIME$.   
\end{proof}

Finally, we consider the case where we have an arbitrary number of parameters. From now on, we start making use of complexity classes involving the \emph{theory of the reals}. 
While the polynomial hierarchy can be understood as finding values for Boolean-valued formulas, the theory of the reals considers real-valued formulas, yielding its own complexity classes. In particular, $ \forall\Reals$ consists of deciding if a formula is true for all possible real valuations of its variables. For an integer $k$, the class $\Pi_k^{\Reals}$ designates the problem of deciding formulas starting with $\forall$ and with $k-1$ quantifier alternations. We give more details on this theory 
\iffullpaper 
in the appendix.
\else 
in \textcolor{red}{CITE}.
\fi

 \begin{restatable}{theorem}{checkreachgen}
    \label{prop:check-reach-gen}
    The following problem is in $\forall\Reals$: 
    Given a \param-MDP $\System = ( \setpos, \legal, X,  \trans, \val,(\funobs_{\ag})_{ \ag\in \Ag}, \constraints)$, and a reachability objective $R$, 
    determine whether for all  $\compvar$-bounded strategies and all well-defined parameters $X$, $R$ can be reached with probability $\bowtie d$. 
\end{restatable}


We can now take a system $\CGS$,  and consider the model checking of \PATL under different definitions of robustness. 
We start with an example, before considering the case of $\varepsilon$-robustness, using a subcase we later build upon.

\begin{example}[River (cont.) ] 
\label{ex:river3}
Going back to our running example, let $\compvar$ denote a memory bound. 
The following formula expresses that a coalition $\coalition$ has a strategy to guarantee, with probability at least $0.5$, that the river quality will be high in the next state: $\coop{C}^{\geq 0.5}\X high$.

In the system $\CGS_{river}$ and starting from state $q_{normal}$, this formula is true with memory at most $\compvar$ for any non-empty coalition, that is, $\CGS_{river}, q_{normal} \models_{\compvar} \bigwedge_{\coalition\neq\emptyset}\coop{\coalition}^{\geq 0.5}\X high$.
However, this is not the case in  the modified system $\CGS_{river}'$ (Figure \ref{fig:riverpert}), 
 since a company alone cannot guarantee to reach $q_{high}$ in the next state with enough probability.  Thus,
$$\CGS_{river}', 
q_{normal} \not\models_{\compvar} \bigwedge_{\coalition\neq\emptyset}\coop{\coalition}^{\geq 0.5}\X high$$


In all realisability of $\CGS_{0.05,river}$, both companies should cooperate to ensure the river will be clean in the next state with a probability at least $0.95$. Thus, $$\CGS_{0.05,river}[\Vval], q_{normal}  \models_{\compvar} \coop{\Ag}^{\geq 0.95}\X high$$
for any valuation $\Vval$.
\end{example}

In the following, a formula $\varphi$ has a chaining of coalition operators if it contains a sub-formula $\coop{\coalition}^{\bowtie d} \psi$, where $\psi$ itself contains a coalition operator. We use such formulas as a base case on which we build our decision procedure for the complete logic. 

\begin{restatable}{lemma}{robpatlunit}\label{prop:rob-patl-unit}
     The following problem is in $\NP$: Given an $\varepsilon$-approximat\-ed system $\System = ( \setpos, \legal, X,  \trans, \val,(\funobs_{\ag})_{ \ag\in \Ag}, \constraints)$, model check $\varepsilon$-robustly a  \PATL  formula $\varphi$ without chaining of coalition operators, and using deterministic  $\compvar$-bounded strategies for the coalitions. 
\end{restatable}

We then conclude on the robust model checking problem:

\begin{restatable}{theorem}{robpatl}\label{prop:rob-patl}
    The following problem is in $\NP^{\NP} = \Sigma_2^{\PTIME}$: Given an $\varepsilon$-approximated system $\System = ( \setpos, \legal, X,  \trans, \val,(\funobs_{\ag})_{ \ag\in \Ag}, \constraints)$, model check $\varepsilon$-robustly a  \PATL  formula $\varphi$ using deterministic  $\compvar$-bounded strategies for the coalitions. 
\end{restatable}



A result similar to Lemma~\ref{prop:rob-patl-unit} holds for parametric systems with a fixed number of parameters, but is in $\NPInter$:

\begin{restatable}{lemma}{patlboundunit}\label{prop:patl-bound-unit}
     The following problem is in $\NP$: Given a \param-system $\System = ( \setpos, \legal, X,  \trans, \val,(\funobs_{\ag})_{ \ag\in \Ag}, C)$, where the number of parameters is less than a fixed $m$, model check parametrically a \PATL  formula $\varphi$ with no chaining of coalition operators, and using deterministic $\compvar$-bounded strategies for the coalitions. 
\end{restatable}

We can then conclude as in Theorem~\ref{prop:rob-patl}. 

\begin{restatable}{theorem}{robpatlbound}\label{prop:rob-patl-bound}
    The following problem is in $\NP^{\NP} = \Sigma_2^{\PTIME}$: 
    Given a \param-system $\System = ( \setpos, \legal, X,  \trans, \val, \allowbreak (\funobs_{\ag})_{ \ag\in \Ag}, C)$, where the number of parameters is fixed and less than $m$, 
    model check parametrically a \PATL formula
    using deterministic $\compvar$-bounded strategies for the coalitions. 
\end{restatable}

We now consider the case of arbitrary parameters:


\begin{restatable}{theorem}{robpatlgen}\label{prop:rob-patl-gen}
      The following problem is in $\exists\forall\exists\Reals = \Sigma_3^{\Reals}$: Given system $\System = ( \setpos, \legal, X,  \trans, \allowbreak  \val,(\funobs_{\ag})_{ \ag\in \Ag}, \allowbreak C)$ and a memory bound $\compvar\in\mathbb{N}_+$, model check parametrically a  \PATL  formula $\varphi$ using randomized $\compvar$-bounded strategies. 
\end{restatable}

\balance
\section{Discussion on Application}
 \label{sec:application}



Cyber-physical systems provide a natural domain where robust strategic
reasoning under partial observability is indispensable \citep{tushar2023survey}. Examples
include smart cities, transportation infrastructures, and industrial
automation, where software controllers interact with a physical
environment whose dynamics are only approximately known, and where
failures or attacks may perturb the expected behaviour of the system.
In such settings, 
%
systems need to be robust, despite noise and disturbances, to function efficiently and avoid the exploitation of potential risks and threats from attackers \citep{Quanyan2015}. 


A representative scenario is that of green buildings in a smart city.
Buildings communicate and coordinate to share resources such as energy
and water efficiently ~\citep{tushar2021peer}. Each building has private information (e.g.,
resource levels or user preferences) and local objectives (e.g.,
minimising costs or contributing resources to the community), while the
overall system should maintain global quality-of-service and safety
constraints~\citep{tushar2023survey}. 
%
This cooperative behaviour can be captured by PATL
specifications. For instance, the formula
 $ \langle\!\langle C\rangle\!\rangle_{\ge 0.8}
  \mathbf{G}\Bigl(\bigwedge_{a \in C} \neg \mathit{energyLow}_a\Bigr)$
expresses that a coalition $C$ of buildings has an observation-based,
bounded-memory strategy to ensure, with probability at least $0.8$,
that none of its members reaches an unacceptably low energy level.
Our robust semantics then asks whether this guarantee holds for \emph{all}
instantiations of the transition probabilities within the prescribed
uncertainty set. This captures, for example, deviations in the
efficiency of solar panels, inaccuracies in predicted heating demand,
or model updates obtained from new data.

%

%


The same framework also supports security-oriented specifications for
cyber-physical infrastructures \citep{Quanyan2015}. Modelling the interaction between
attackers and defenders as a stochastic MAS, one can express properties
such as
 $ \neg\langle\!\langle \{a\}\rangle\!\rangle_{\ge 0.2}
  \mathbf{F}\bigl(\mathit{access}_a \wedge \neg\mathit{securityCheck}_a\bigr)$,
which states that an attacker $a$ does \emph{not} have a strategy to
gain unauthorized access with probability at least $0.2$. Robust model
checking then verifies that this remains true even if the probabilities
of successful attacks or detection change. 

From a modeling perspective, our robust stochastic MAS can be viewed
as a multi-agent generalization of classical uncertain models such as
interval and parametric MDPs, enhanced with coalitional strategic
operators and observation-based, bounded-memory strategies. The
complexity results in Table~1 show that, under these general conditions,
robust verification is still decidable with complexities that are
compatible with existing verification technology. This suggests that
our framework can serve as a semantic foundation for future
tool-supported analysis of robust strategies in realistic
cyber-physical applications.
\medskip

\section{Conclusion}\label{sec:conclusion}

This paper addresses the problem of verifying the robustness of strategies for agents acting in stochastic MAS whose transition probabilities are uncertain. Focusing on observation-based, bounded-memory strategies represented by finite automata, we have
introduced a robust variant of the model-checking problem for $\PATL$ and established upper bounds on its complexity under three different notions of perturbation: $\varepsilon$-bounded deviations of each transition, parametric models with a fixed number of parameters, and parametric models with an unbounded number of parameters. 
On the quantitative side, our results clarify the computational cost of verifiable robustness in MAS. In particular, we show that robust strategic reasoning remains decidable even in the presence of coalitions, imperfect information, and parametric uncertainty, while highlighting the sharp complexity jump induced by unboundedly many parameters. Conceptually, our work connects several previously separate lines of
research: verification of stochastic MAS in strategic logics, interval
and parametric models for uncertain probabilities, and robust strategy
synthesis. 

We left open the case of lower bounds; while some are unrealistically high, we think it would be possible to get lower bounds on the robust model checking of \PATL.
As future work, we aim to consider the problem of computing the robustness level of a given strategy, that is, the maximal perturbation under which the strategy is still successful. This would open the question of finding the most robust strategy in a stochastic MAS.  
\begin{acks}
This research has been supported by the DFG Project POMPOM (KA 1462/6-1) and ANR-22-CE48-0012.
\end{acks}

\bibliographystyle{ACM-Reference-Format} 
\bibliography{ref}

\iffullpaper 
\clearpage
\appendix
\label{appendix}
\section{
Appendix}

\subsection{Proof of Proposition 
\ref{prop:strat_gen}}\label{appendix:strat_gen}

\propstratgen* 
\begin{proof}
    Given a $k$-sequential strategy $\sigma:\funobs_{\ag}(O^k)\rightarrow\distribution(\Act)$, we build the scheduler $\sigma'$ as $(\{O\times \{\varnothing\}\}^{k-1},act,\Delta,start)$, where 
$\Delta(q,o) = \last^{k-2}(q)\cdot o$ where $\last^{k-2}(q)$ is the sequence of the $k-2$ last elements of $q$, 
$act(q,o) = \sigma(q\cdot o)$ if defined and anything otherwise,  
and $start(o) = \{\varnothing\}^{k-1}\cdot o$. This scheduler has the output $\sigma(o)$ on any initial state $o$, and by induction, if $\sigma$ and $\sigma'$ have the same output after a sequence $q\in O^*$, they have the same output on $q\cdot o$. 

For the other side, let us consider the system 
$\System_g$ on Figure~\ref{fig:mdp_gen}. The only atomic propositions are on states $q_1$ and $q_5$ which are labeled with $blue$, while $q_2$ and $q_6$ are labeled with $red$. In most states, the two actions $r$ and $b$ yield the same result, except in $q_4$ where $b$ leads to a blue sink and $r$ leads to a red sink. We now consider the \LTL formula $\varphi_g = \G ((blue\Rightarrow \G \F blue) \wedge (red\Rightarrow \G \F red))$. This formula states that if a blue state is visited even once, then infinitely many blue states are visited, and the same holds for red states. We consider the scheduler of Figure~\ref{fig:control_gen}, playing action $b$ in state $q_b$ and action $r$ in state $q_r$. After any history on $\System_g$ it will be in state $q_b$ (respectively $q_r$) if and only if the blue state $q_1$ (respectively the red state $q_2$) has been visited, and will play the corresponding action, ensuring $\varphi_g$ with probability $1$. On the other hand, given any memory-bound $k$, for any $k$-sequential strategy, all histories that stay in $q_3$ for $k$ steps should yield the same action in $q_4$, whether $q_1$ or $q_2$ was visited, which means $\varphi_g$ cannot be satisfied with probability $1$.

\end{proof}
We highlight that our example uses a deterministic and perfect-information scheduler-based strategy. ~\citep{DBLP:journals/corr/abs-1006-1404} details how randomization leads to even more general strategies, hence our choice to call our strategies ``general'', or to simply refer to them as ``strategies''. We note that the two definitions coincide on infinite-memory strategies.




\subsection{Remark on a possible $\Sigma_2^{\PTIME}$-hardness result (follow-up from Theorem~\ref{prop:rob-patl})}
We show that $\varepsilon$-robust model-checking a \PATL formula on a system can be done in $\Sigma_2^{\PTIME}$ in Theorem~\ref{prop:rob-patl}. An approach to show $\Sigma_2^{\PTIME}$-hardness would be to encode $2\mathbf{QBF}$ into it, which consists in the following problem: decide if a formula of the form $\exists X\forall Y B(X,Y)$ holds, where $B$ is a Boolean combination of atoms in the sets $X$ and $Y$. 

It is simple to encode such a formula into a system with two agents (named $a_{\exists}$ and $a_{\forall}$) whenever $B$ is a CNF formula: start in a state associated to a random clause of $V$, and make agent $a_{\exists}$ successively choose valuations for variables in $X$ (states are labelled by the variable to choose, so this agent cannot play differently depending on the clause). If any of the valuations satisfies the clause, go to some winning state $s_{\top}$, otherwise proceed similarly for agent $a_{\forall}$ on variables in $Y$. We then check the formula $\coop{a_{\exists}}^{=1} \F s_{\top}$: agent $a_{\exists}$ has a (memoryless) strategy surely reaching $s_{\top}$ if and only if formula $\exists X\forall Y B(X,Y)$ holds. 

The problem with this approach is that $2\mathbf{QBF}$ on CNF formulas is in $\NP$: it is possible to guess a valuation for $X$, and then assume in every clause independently that $Y$ takes the worst possible value, which can be checked in polynomial time. 

When considering an arbitrary Boolean formula $B$, the construction above does not work anymore. It is possible to have agent $a_{\exists}$ decide on a (memoryless and observation-based) strategy assigning valuations to variables in $X$, and then go through $B$, where randomness selects either the right or left subformula when meeting a $\vee$, and letting agent $a_{\forall}$ choose the next subformula when meeting a $\wedge$, until meeting a variable. The agent associated to this variable decides on a valuation, if this valuation satisfies the subformula then the system goes to state $s_{\top}$. If agent $a_{\forall}$ had imperfect information, $\coop{a_{\exists}}^{>0} \F s_{\top}$ would hold if and only if $\exists X\forall Y B(X,Y)$ holds, but this is not the case, and $a_{\forall}$ can choose different valuations for the same variable in different parts of the formula. 

There is no obvious way to go around this problem using \PATL as defined. While small modifications of the problem considered could lead to a $\Sigma_2^{\PTIME}$-hardness result, they would not easily lead to $\Sigma_2^{\PTIME}$-completeness. 

If instead of \PATL we used a variant of $\ATL^*$ it would be possible to have agent $a_{\exists}$ display all their chosen valuations, then have agent $a_{\forall}$ do the same, while stating that once a valuation is chosen, the same valuation should always be chosen ($a\Rightarrow \G a$). This would, however, prevent us from using our model checking technique, since we couldn't easily transform an $\ATL^*$ formula into reachability and invariance objectives. 

Another approach would be to have agent $a_{\forall}$ also use observation-based memoryless strategies, so that the choice of valuations couldn't depend on the position in the formula. This would change the decision procedure: we could not use Lemma~\ref{prop:rob-patl-unit} and Theorem~\ref{prop:check-reach-eps} anymore, since they allow arbitrary strategies. On a POMDP, deciding whether a reachability formula is satisfied by all memoryless deterministic strategies with at least a given probability is $\coNP$-complete (as the dual of the existential case in Proposition 15 of \citep{DBLP:journals/jcss/JungesK0W21}, that we recall as Proposition~\ref{prop:pomdp-mem}), while it is polynomial for MDPs, this would in turn greatly raise the complexity of our proposed model checking technique.

\begin{proposition}[From~\citep{DBLP:journals/jcss/JungesK0W21}]
\label{prop:pomdp-mem}
    For pMDPs with fixed parameter, $\exists \mathrm{REACH}^{\bowtie}_{\star}$ is in \NP.
\end{proposition}

\subsection{Polynomial Hierarchy and Theory of the Reals (background for  
Theorem~\ref{prop:check-reach-fixed})}
\label{appendix:th_reals}

\paragraph{Polynomial Hierarchy. }
The classical complexity class $\NP$ can be defined as the set of decision problems for which the \texttt{yes} instance can be \emph{verified} in polynomial time using a polynomial-size witness. The $\mathbf{SAT}$ problem is a canonical $\NP$-complete problem. $\coNP$ is the dual complexity class, where the \texttt{no} instance can be verified similarly. 
This has been expanded to a hierarchy using \emph{oracles}. Given two complexity classes, $\mathbf{A}$ and $\mathbf{B}$, the class of problems that belong to class $\mathbf{A}$ when given an oracle for class $\mathbf{B}$ is denoted $\mathbf{A}^{\mathbf{B}}$. When considering the \emph{polynomial hierarchy}~\cite{DBLP:journals/tcs/Stockmeyer76}, $\NP$ is denoted as $\Sigma_1^{\PTIME}$, and $\coNP$ as $\Pi_1^{\PTIME}$. When, for some $k$ a problem is in $\PTIME$ with a $\Sigma_k^{\PTIME}$ oracle (or equivalently a $\Pi_k^{\PTIME}$ oracle) it is in $\Delta_{k+1}^{\PTIME}$. Similarly, the complexity class $\NP^{\Sigma_k^{\PTIME}}$ is denoted $\Sigma_{k+1}^{\PTIME}$, and  $\coNP^{\Pi_k^{\PTIME}}$ is denoted $\Pi_{k+1}^{\PTIME}$. These complexity classes are widely used in model checking, in particular model checking $\ATL$ with imperfect information and imperfect memory is $\Delta_2^p$-complete~\citep{DBLP:conf/eumas/JamrogaD06} and a similar result exists for Nat\ATL with memoryless strategies~\citep{natStrategy}.

While the polynomial hierarchy can be understood as finding values for Boolean-valued formulas, the theory of the reals considers real-valued formulas.

The \emph{existential theory of the reals}, $\exists\Reals$, is the class of problems that can be polynomially reduced to deciding the truth of a first-order formula
$\Phi \equiv \exists x_1 \dots \exists x_n\ P(x_1,\dots,x_n)$
where $x_i$ are interpreted over the reals $\Reals$, and $P$ is a 
Boolean function of atomic predicates of the form $f_i(x_1,\dots,x_n) \ge 0$ or $f_i(x_1,\dots,x_n) > 0$, with each $f_i$ being a polynomial with rational coefficients. 
$\exists\Reals$ is known to be in \PSPACE and \NP-hard~\citep{Canny88PSPACE}.
Many problems involving both partial observation and probabilities are related to this class, which has led to a recent surge of interest. In particular, finding a probabilistic strategy reaching a target with at least a given probability and using some bounded memory is $\exists\Reals$-complete~\citep{DBLP:phd/dnb/Junges20}. See also~\citep{schaefer2024existential} for a survey on $\exists\Reals$. 
Similar problems with more alternation of quantifiers also create a hierarchy. For an integer $k$, $\Sigma_k^{\Reals}$ designates the problem of deciding formulas starting with $\exists$ and with $k-1$ quantifier alternations. Hence $\Sigma_1^{\Reals} = \exists\Reals$, and $\Sigma_2^{\Reals} = \exists\forall\Reals$. The classes $\Pi_k^{\Reals}$ are the dual, for formulas starting with a $\forall$ quantifier. While several problems are in $\Sigma_2^{\Reals}$~\citep{DBLP:journals/mst/SchaeferS24}, finding complete problems has been challenging, but the compact escape problem has been shown to be such a complete problem~\citep{DBLP:conf/mfcs/DCostaLNO021}. Interestingly, \citep{DBLP:journals/mst/SchaeferS24} and~\citep{DBLP:journals/dcg/DobbinsKMR23} consider that the three main settings leading to a jump from the first to the second level of the hierarchy are universal extension, imprecision, and robustness. We are in the third, since we consider multiple agents with imperfect information that need to follow robust strategies.

\subsection{Proof of Theorem \ref{prop:check-reach-eps}}

The proof of Theorem \ref{prop:check-reach-eps} will use the following results: 

\begin{proposition}[From ~\citep{DBLP:journals/ipl/ChenHK13}]
\label{prop:coincide}
    The UMC and IMDP semantics coincide when considering reachability objectives. 
\end{proposition}

\begin{proposition}[From ~\citep{DBLP:journals/ipl/ChenHK13}]
\label{prop:idtmc}
    Given an $\varepsilon$-approximated MDP $\mathcal{I}$, computing the parameters and strategy that have maximal (resp. minimal) probability to reach a target is polynomial. 
\end{proposition}

\checkreacheps*
\begin{proof}
    This result stems from using results of~\citep{DBLP:journals/ipl/ChenHK13}, where $\varepsilon$-approximated MDPs are similar to \emph{interval Markov chains under the IMDP semantics}, which are MDPs where at every step, the player chooses a probability distribution for the next transition, as long as it stays within some given bounds. Finding a strategy with maximal (resp. minimal) probability to reach a target can be done in polynomial time in such a system, as we recall in Proposition~\ref{prop:idtmc} below. Hence deciding whether a target can be reached with probability $d$ is also polynomial. 
    We can obtain a similar setting by separating states of our $\varepsilon$-approximated MDP between states where we choose the next action, without restriction since we are allowed to randomize the action we take, and states where we decide the probability distribution of the next state given the action that has been chosen, within the $\varepsilon$ bound. 
    Proposition 1 in~\citep{DBLP:journals/ipl/ChenHK13}, that we recall and rephrase as Proposition~\ref{prop:coincide} below, shows that choosing the transition probability at every step (the IMDP semantics) or fixing them from the start (that they call the UMC semantics) is equivalent for reachability properties, which allows us to use Proposition~\ref{prop:idtmc}. 
    Our problem is the dual of Proposition~\ref{prop:idtmc}, and so we can use it directly: on an $\varepsilon$-approximated MDP $\System $, all strategies satisfy for all valuations a reachability objective with probability $\bowtie d$ ($\bowtie\in\{<,>,\geq,\leq\}$) iff it does not hold that there exists a strategy and a valuation such that the same reachability objective holds with probability $\bar{\bowtie} d$, where $\bar{<}$ is $\geq$, $\bar{>}$ is $\leq$, $\bar{\geq}$ is $<$, and $\bar{\leq}$ is $>$. 
\end{proof}

\subsection{Proof of Theorem \ref{prop:check-reach-gen}}

\checkreachgen* 

\begin{proof}
 When the set of parameters is fixed and instantiated, we obtain an MDP. In an MDP if there exists a strategy that does not reach $R$, it can be assumed to be memoryless (by Lemma 10.102 of~\citep{DBLP:books/daglib/0020348}). Hence, to solve this problem, it is sufficient to decide if for all memoryless strategies and all well-defined parameters, $R$ is reached. This can be expressed using the universal theory of the reals. 
\end{proof}

\subsection{Proof of Lemma \ref{prop:rob-patl-unit}}

\robpatlunit*

\begin{proof}
We first guess a deterministic strategy for the coalition. Its size is polynomial in the memory bound fixed in the coalition operator. We fix this strategy, and the coalition is now empty: we obtain an $\varepsilon$-approximated MDP $\System'$ in which only agents that do not belong to the coalition remain. We use classical techniques from Propositions 5.1 and 5.2 of~\citep{AlurHK02} to transform until and release operators into reachability. Checking whether the agents in the \param-MDP $\System'$ have no way to falsify a reachability property consists in checking whether for all strategies and well-defined $\varepsilon$-perturbations with a bound $m$, this reachability objective holds. By Theorem~\ref{prop:check-reach-eps} this is polynomial, hence the problem is in $\NP$. 
\end{proof}

\subsection{Proof of Theorem \ref{prop:rob-patl}}

\robpatl* 
\begin{proof}

We start with membership, giving a way to robustly model check such a system and formula. 
We introduce Boolean variables $a_1,\ldots, a_n$ for each subformula $\coop{C_i}^{\bowtie d_i}\varphi_i$ of $\varphi$, with $i\in [1,n]$. In every subformula $\varphi_i$, we replace the cooperation operators by their respective propositional variable. Assuming we are given an $\NP$ oracle, we show that the problem is in $\NP$: indeed, we can guess whether every $a_1,\ldots, a_n$ holds, and verify the correctness of this guess in polynomial time using our $\NP$ oracle on  $\coop{C_i}^{\bowtie d}\varphi_i$ thanks to Lemma~\ref{prop:rob-patl-unit}.
\end{proof}

\subsection{Proof of Lemma \ref{prop:patl-bound-unit}}

\patlboundunit*

\begin{proof}
After having guessed a deterministic strategy for the coalition, and fixing this strategy, we obtain a \param-MDP. We again transform until and release operators into reachability using~\citep{AlurHK02}. Checking on this \param-MDP that for all strategies and a fixed amount of parameters, a reachability objective holds is in $\NPInter$ by Theorem~\ref{prop:check-reach-fixed}, hence the problem is in $\NP^{\NPInter}$ which is $\NP$~\citep{DBLP:journals/jcss/Schoning83}.
\end{proof}


\paragraph{Theorem \ref{prop:rob-patl-bound}.} 
The proof is obtained using Lemma \ref{prop:patl-bound-unit} the exact same as in Theorem  \ref{prop:rob-patl}, and we omit it.

\subsection{Proof of Theorem \ref{prop:rob-patl-gen}}

\robpatlgen* 
\begin{proof}
    Let $\varphi$ be a \PATL formula 
    and $\System = ( \setpos, \legal, X,  \trans, \val,(\funobs_{\ag})_{ \ag\in \Ag}, \allowbreak C)$ be a parametric system. For the $n$ cooperation subformulas $\coop{C_i}^{\bowtie d_i}\varphi_i$, we introduce Boolean variables $a_1,\ldots, a_n$ for $i\in [1,n]$. In every subformula $\varphi_i$, we replace cooperation operators by their respective propositional variable. Deciding if $\varphi$ holds is the same as deciding if there is a consistent valuation for these $a_1,\ldots, a_n$ and bounded-memory $\sigma_1,\ldots,\sigma_n$ strategies for each subformula, such that for all well-defined probability transitions $P_{T}$, in every subformula $i$, for all strategies $\mu_i$ of the agents not in the coalition formula, $\varphi_i$ holds iff $a_i$ is true. For the negative case, we duplicate it into formulas $\varphi'_i$, associated to $\sigma'_i$, $\mu'_i$ and to a specific $P_{T,i}$ witnessing the coalition cannot achieve $\varphi_i$. Since once bounded memory strategies and the transition probabilities are fixed, we have a classical MDP with threshold objective, we can freely assume (by Lemma 10.102 of~\citep{DBLP:books/daglib/0020348}) that strategies $\mu_i$ are deterministic and memoryless, and thus can be represented with a number of variables polynomial in $\System$. Formally, we consider formula $\exists a_1,\ldots, a_n. \bigwedge_{i\in[1,n]}(a_i\Rightarrow \exists \sigma_{i}. \forall P_{T}. \forall \mu_i. \varphi_i) \wedge (\neg a_i\Rightarrow \forall \sigma'_i. \exists P'_{T,i}. \exists \mu'_i. \varphi'_i)$ where every $\varphi_i$ with $i\in[1,n]$ states that when following strategies $\sigma_{i}$ and $\mu_i$ on the system $\System$ with transition probabilities $P_{T,i}$ within the $\varepsilon$-robustness, the valuation of each $\phi_i$ is $a_i$. Since none of the quantified variables are shared, we can write it as the following formula in $\Sigma_3^{\Reals}$:  
    \begin{align*}
    \exists a_1,\ldots, a_n. \exists \sigma_1,\ldots,\sigma_n. \forall P_{T}. \forall \mu_1,\ldots,\mu_n. \\ \forall \sigma'_1,\ldots,\sigma'_n.\exists P'_{T,1},\ldots, P'_{T,n}. \exists \mu'_1,\ldots \mu'_n. \\ \bigwedge_{i\in[1,n]}(a_i\Rightarrow  \varphi_i) \wedge (\neg a_i\Rightarrow   \varphi'_i)    
    \end{align*}
\end{proof}  
\fi


\end{document}